\documentclass[11pt,a4paper]{amsart}

\usepackage{amsthm}
\usepackage{amsfonts}
\usepackage{amsmath}
\usepackage{amssymb}
\usepackage{mathrsfs}
\usepackage{ifthen}

  \theoremstyle{definition}
  \newtheorem{definition}{Definition}[section]
  \renewcommand{\thedefinition}{\arabic{section}.\arabic{definition}}
  
  \theoremstyle{definition}
  \newtheorem{remark}[definition]{Remark}

  \theoremstyle{plain}
  \newtheorem{theorem}[definition]{Theorem}

  \theoremstyle{plain}
  \newtheorem{lemma}[definition]{Lemma}

  \theoremstyle{plain}
  \newtheorem{proposition}[definition]{Proposition}

  \theoremstyle{plain}
  \newtheorem{corollary}[definition]{Corollary}

\theoremstyle{definition}
  \newtheorem{assumptions}[definition]{Assumption}

\theoremstyle{plain}
  \newtheorem*{results}{Main Results}

  \renewcommand{\theequation}{\arabic{section}.\arabic{equation}}
  \newcommand{\BbbR}{\ensuremath{\mathbb {R}}}  
  \newcommand{\BbbN}{{\mathbb {N}}}  
  \newcommand{\BbbP}{{\mathbb {P}}}

  \newcommand{\expval}[2][]{\ifthenelse{\equal{#1}{}}{\mathbb{E}\left[#2\right]}{\mathbb{E}\left[#2|#1\right]}}   
  \newcommand{\Levy}{L\'{e}vy}
  \newcommand{\stochexp}[2][]{\ensuremath{\mathcal{E}^{#1}\left(#2\right)}}

  \newcommand{\sbullet}{\;\begin{picture}(1,1)(0,-3)\circle*{3}\end{picture}\; }
  \newcommand{\C}{\mathscr{C}}
  \newcommand{\calC}{\mathcal{C}}
\begin{document}

  \title[Power Utility Maximization in Exponential L\'{e}vy Models]{Power Utility Maximization in Exponential L\'{e}vy Models: Convergence of Discrete-Time to Continuous-Time Maximizers}
  \author{Johannes P.\ Temme}
  \address{Fakult\"at f\"ur Mathematik, Universit\"at Wien\endgraf Nordbergstra\ss e 15\\ 1090 Wien, Austria}
  \email{johannes.temme@univie.ac.at}
  \date{\today}  
  \thanks{The author gratefully acknowledges financial support from the Austrian Science
Fund (FWF) under grant P19456. The author thanks two anonymous referees for their valuable input. He also thanks Johannes Muhle-Karbe and Mathias Beiglb\"ock for fruitful discussions and comments on the draft. The final publication is available at \texttt{springerlink.com}}

  \subjclass[2000]{Primary 91B28, 91B16, secondary 60G51} 
  \keywords{utility maximization, power utility, exponential \Levy\ process, discretization}

  \begin{abstract}
    \noindent We consider power utility maximization of terminal wealth in a 1-dimensional continuous-time exponential \Levy\ model with finite time horizon. We discretize the model by restricting portfolio adjustments to an equidistant discrete time grid. Under minimal assumptions we prove convergence of the optimal discrete-time strategies to the continuous-time counterpart. In addition, we provide and compare qualitative properties of the discrete-time and continuous-time optimizers. 

  \end{abstract}
\maketitle


\section{Introduction}

  In exponential \Levy\ models the stock is given -- just as its name suggests -- as the (stochastic) exponential of a \Levy\ process. Continuous-time exponential \Levy\ models and their discrete-time counterparts share the defining property of independent and identically distributed logarithmic returns. Exponential \Levy\ models have become widely popular in the last decade since they are analytically tractable and provide a reasonable approximation to financial data.

  We consider maximization of expected utility from terminal wealth with finite time horizon for an investor e\-quipped with a \emph{power utility function} $U(x)=(1-p)^{-1}x^{1-p}$, $p>0$, in the framework of one bond (normalized to 1) and one stock which follows an exponential \Levy\ process. The problem of maximizing expected power utility from terminal wealth was first studied by Merton (\cite{merton69lifetime}) for Brownian motion with drift and by Samuelson (\cite{samuelson69lifetime}) for the discrete-time analogon in an $N$-period model. The optimal portfolio selection for general time-continuous exponential \Levy\ processes was investigated in  \cite{framstad_oksendal99optimal}, \cite{kallsen00optimal} and \cite{benth_karlsen01optimal}. Nutz (\cite{nutz09power}, for $p\neq 1$), as well as Kardaras (\cite{kardaras09no}, for $p=1$) have recently shown under minimal assumptions that the optimal fraction\footnote{Following a usual practice for power utiltiy functions, we note that we interpret in this exposition trading strategies as the \emph{fraction} of current wealth invested in the stock, rather than the \emph{amount of shares} the agent holds. Thus, for a given trading strategy $\pi_t$ the agent invests the fraction $\pi_t$ of her current wealth in the stock, and the fraction $1-\pi_t$ of current wealth in the bond. } $\pi^*$ of wealth invested in the continuous-time exponential \Levy\ model is \emph{constant} and given as the maximizer of a deterministic concave function $g$, which is defined in terms of the \Levy\ triplet (also cf., e.g., \cite{goll_kallsen00} for sufficient conditions in a general setting, and \cite{Fo89} for sufficient conditions on the optimality of constant trading strategies). Similarly, the optimal trading strategy for an $N$-period exponential \Levy\ model is to invest a \emph{constant} fraction $\pi^*_N$ of current wealth in each step. This stems from the fact that the logarithmic return of the stock in each period is i.i.d.\ and consequently the $N$-period model becomes an iterated 1-period model. The optimal fraction $\pi_N^*$ is given as the maximizer of a deterministic concave function $g^N$, where $g^N$ only depends on the evolution of the \Levy\ process in the first step. 
  \medskip

  In this exposition we assume that a continuous-time exponential \Levy\ model is given. We define the discrete-time $N$-period exponential \Levy\ model as the restriction of the continuous-time model to discrete time instants of distance $\frac{T}{N}$, where $T$ is the finite time horizon. In particular, the discrete logarithmic returns are given by i.i.d.\ random variables. Hence, trading in this discrete-time model amounts to trading in the original con\-tin\-u\-ous-time model with the restriction that the portfolio can only be adjusted at given time instant, i.e., if $t=\frac{kT}{N}$, $k=0,\ldots,N$. This restriction also involves a discretization gap in the sense that the set of admissible trading strategies for the $N$-period model is in general strictly smaller than the respective set of the continuous-time model. For this reason the optimal fraction $\pi^*$ invested in the stock of the continuous-time model might in general be \emph{not} admissible for the discrete-time models. This somewhat surprising consequence was already mentioned by Rogers (\cite{rogers01relaxed}) in the case of the classical Merton problem, which we consider in more detail below.
   \medskip
  
  If the \Levy\ process has a non-zero Brownian motion part or allows both positive and negative jumps, it is easily seen that the interval of admissible constant trading strategies of all $N$-period models is given by $[0,1]$. This means that the agent is allowed to invest in each period the fraction of 0\% to 100\% of current wealth in the stock. Short-selling and investing more than 100 \% of current wealth is prohibited since negative wealth is not allowed by power utility investors. In the other (less important) case of a pure jump process allowing \emph{either} only positive \emph{or} only negative jumps, the interval of admissible trading strategies for the $N$-period model is more complicated and depends on the number of discretization points $N$. Although not technically necessary, in order to simplify the further discussion we will assume throughout this exposition that the \Levy\ process has a non-zero Brownian motion part or allows both positive and negative jumps. Hence, the interval of admissible constant $N$-period trading strategies is in particular given by $[0,1]$. We comment separately on the case of a pure jump process with either only positive or only negative jumps.
  \medskip

  The objective of this exposition is to answer the following two questions: 
  \begin{enumerate}
   \item[1.)] Do the optimal discrete-time strategies $\pi_N^*$ converge to the optimal continuous-time strategy $\pi^*$ as $N\to\infty$, i.e., as the number of discretization points increases?
   \item[2.)] How does the sign of the drift of the \Levy\ process and the risk aversion of the power utility funcion affect the optimal continuous-time strategy $\pi^*$ and the discrete-time strategies $\pi_N^*$? 
  \end{enumerate}

  Clearly, a necessary assumption for the convergence of $\pi_N^*\in[0,1]$ to $\pi^*$ is $\pi^*\in[0,1]$. But $\pi^*\in[0,1]$ generally fails to hold already in the Black-Scholes model due to the previously mentioned discretization gap: in the classical Merton problem the \Levy\ process is given by $L_t=\mu t+\sigma B_t$, where $B$ denotes Brownian motion. Let $U(x)=(1-p)^{-1}x^{1-p}$ with $p>0$, then the optimal continuous-time strategy is given by the constant $\pi^*=\frac{\mu}{\sigma^2p}$, cf.\ \cite{merton69lifetime}. Thus, $\pi^*\notin[0,1]$ if e.g.\ $\mu$ is chosen large enough and $\pi_N^*$ will not converge to $\pi^*$.

  To cope with this problem we introduce in the continuous-time exponential \Levy\ model the exogenous portfolio constraint $\C=[0,1]$. I.e., in this constrained continuous-time model the agent may only invest between $0$\% and $100$\% of her current wealth in the stock, which is the same restriction as in the $N$-period models. The optimal continuous-time strategy $\pi_\C^*$ of the constrained model can be found similarly as in the unconstrained case as the maximium of the same function $g$ on $\C=[0,1]$. In particular, $\pi^*\in[0,1]$ implies $\pi^*=\pi^*_\C$.
  
  Coming back to our two questions we see that we have to refine the first one due to the discretization gap and rather ask 
  \begin{enumerate}
   \item[1.)'] Do the optimal discrete-time strategies $\pi_N^*$ converge to the optimal strategy $\pi^*_\C$ of the \emph{constrained} continuous-time model as $N\to\infty$, i.e., as the number of discretization points increases?   
  \end{enumerate}
  In order to respond to both questions we observe that we can access $\pi^*_\C$ and $\pi_N^*$ as the maximizer of the previously mentioned deterministic and concave functions $g|_{[0,1]}$ and $g^N$, respectively. A major part in answering the questions is done in Theorem \ref{thm:conv_optimal_strategies} below, which shows that $g^N$ converges uniformly on $[0,1]$ to $g$ as $N\to\infty$. Our main results can then be summarized as follows:
  
  \begin{results}[Theorem \ref{thm:conv_optimal_strategies}, Corollaries \ref{cor:conv_utility} and \ref{cor:conv_wealth}, Proposition \ref{prop:levy_order_of_strategies}]\rule{0pt}{0pt}
  Let $\pi^*$, $\pi_N^*$ and $\pi^*_\C$ denote the optimal fraction of current wealth invested in the stock for maximizing expected \emph{power utility} $U(x)=(1-p)^{-1}x^{1-p}$, $p>0$, from terminal wealth in the \emph{unconstrained} continuous-time exponential \Levy\ model, its $N$-period \emph{discretization} and the \emph{constrained} continuous-time exponential \Levy\ model with exogenous portfolio constraints $\C=[0,1],$ respectively.

  Then 
     \begin{enumerate}
      \item\label{it:mr1} $\lim_{N\to\infty}\pi_N^*=\pi^*_\C$, 
      \item the terminal expected utility of the $N$-period model converges to the terminal expected utility of the constrained continuous-time model as $N\to\infty$, and
      \item\label{it:mr3} if additionally the \Levy\ process is square integrable, then the optimal $N$-period terminal payoffs converge in $L^2(\Omega)$ to the optimal terminal payoff of the constrained continuous-time model as $N\to\infty$.
     \end{enumerate}
  Hence, if $\pi^*\in[0,1]$, i.e., if $\pi^*$ is admissible in the $N$-period models, (\ref{it:mr1})-(\ref{it:mr3}) hold for the unconstrained continuous-time exponential \Levy\ model.

  We also find the following qualitative properties of the optimal unconstrained continuous-time and $N$-period strategies under the assumption that the \Levy\ process is integrable:
    \begin{enumerate}
      \setcounter{enumi}{3}      
      \item \label{it:mr4} $\pi^*, \pi^*_N$ are non-negative (non-positive) if the drift of the \Levy\ process -- i.e., the expectation of the \Levy\ process -- is non-negative (non-positive). Both $|\pi^*|$ and $\pi^*_N\in[0,1]$ are increasing if the relative risk aversion $p$ of $U$ is decreasing.\footnote{Note that $\pi^*$ as well as $\pi^*_N$ depend on $U(x)=(1-p)^{-1}x^{1-p}$ and can be considered as functions of $p$}
    \end{enumerate}
  \end{results}

  In the light of our results, we see that the exogenous portfolio constraint $\C=[0,1]$ is a natural assumption to make the continuous-time model resemble the discrete-time models. As we have already mentioned, $\pi^*\in[0,1]$ implies $\pi^*=\pi^*_\C$. On the other hand, if $\pi^*\notin[0,1]$ then $\pi_\C^*$ and $\pi^*_N$ are easily found: if $\pi^*<0$  then $\pi_N^*=\pi_\C^*=0$, and if $\pi^*>1$ then $\pi_N^*=\pi_\C^*=1$. 
  \medskip

   The approximation of a continuous-time \Levy\ model by discretizations has been covered under different aspects in the literature. Sufficient conditions for the convergence of optimal trading strategies in certain diffusion models and their discrete approximations were found in \cite{He91}. Rogers (\cite{rogers01relaxed}) compared optimal investors of a diffusion model and its discretization in terms of \emph{efficiency} and his work was extended in \cite{BaUr10} to the case of partially available information. Another aspect of discretization in terms of \emph{time-lagged trading} and its impact on utility maximzation was discussed in \cite{RoSt02} for exponential \Levy\ models.

   \medskip
   
  The paper is organized as follows. In Section 1 we specify the continuous-time and discrete-time exponential \Levy\ models, as well as the optimization problem of maximizing expected utility from terminal wealth. We state the necessary assumptions (Assumption \ref{assumptions}) for our results and recall certain findings of \cite{nutz09power}, as well as some properties of stochastic exponentials of \Levy\ processes that we shall frequently use. Section \ref{sec:prop_g} summarizes important analytic properties of the deterministic functions $g$ and $g^N$. In Section \ref{sec:conv_g} we state and prove our main results. We also comment on the extension of Theorem \ref{thm:conv_optimal_strategies} to higher dimensional exponential \Levy\ processes. Certain technical results needed for the convergence of the wealth processes (Corollary \ref{cor:conv_wealth}) are presented in the Appendix.
  
  It is important to mention that our convergence theorem (Theorem \ref{thm:conv_optimal_strategies}) relies on the work of \cite{jacod07asymptotics} and \cite{figueroa08small} on moment asymptotics for \Levy\ processes. We also use results on the finiteness of $g$, that are proved in \cite{nutz09power}. The convergence results in the Appendix are extensions of the results in \cite{kohatsu_protter94euler} on Euler discretizations.

\section{Preliminaries and Assumptions}
  We fix a finite time horizon $T>0$ and a filtered probability space $(\Omega, \mathcal{F}, (\mathcal{F}_t)_{t\geq 0},\BbbP)$ satisfying the usual assumptions of right-continuity and completeness. The exponential \Levy\ model of one bond normalized to $1$ and one stock $S$ with initial value $S_0>0$ is given as the solution of the SDE
  \begin{equation}\label{eq:stock_contLevy}
    S_t=S_0+\int_0^tS_{u-}\;dL_u,
  \end{equation}
  where $(L_t)_t$ is a \Levy\ process with \Levy-Khintchine triplet $(b(h),c,F)$ relative to a continuously differentiable truncation function $h$, with drift $b(h)\in\BbbR$, diffusion $c>0$ and \Levy\ measure $F$ on $\BbbR$. We refer to \cite[II.4]{jacod_shiryaev03limit} for more information on \Levy\ processes.
  \medskip
  
  By \eqref{eq:stock_contLevy}, $S$ is given by the stochastic exponential of $L$, i.e., $S_t=S_0\stochexp{L}_t$. The agent is endowed with initial capital $x_0>0$ and her preferences are given by a power utility function $U(x)=\frac{x^{1-p}}{1-p}$ with $p>0$, where we set throughout this exposition $\frac{x^0}{0}:=\log(x)$, to cover also the case of logarithmic utility. Let $\calC\subseteq\BbbR$ denote an interval of portfolio constraints. Throughout this exposition we are only interested in the unconstrained case $\calC=\BbbR$ and in the exogenous constraint $\calC=\C:=[0,1]$.

  We define the set of all admissible constant trading strategies by the convex set
  \begin{align}
    \mathcal{A}_0&:=\{\pi\in\BbbR|F\left(x\in\BbbR: 1+\pi x<0\right)=0\} \quad\text{if }p\in(0,1),\label{eq:mathcalA}\\
    \mathcal{A}_{0,*}&:=\{\pi\in\BbbR|F\left(x\in\BbbR:1+\pi x\leq 0\right)=0\}\quad\text{if }p\geq 1,\label{eq:mathcalAstar}
  \end{align}
  where $\pi$ denotes the \emph{fraction} of wealth invested in the stock. Since we fix the utility function $U$, we denote the corresponding set of admissible constant trading strategies by $\mathcal{A}$ to simplify further notation, while we still have in mind that we mean either of \eqref{eq:mathcalA}, \eqref{eq:mathcalAstar}.  The requirements in \eqref{eq:mathcalA} and \eqref{eq:mathcalAstar}, respectively, ensure that $x_0\stochexp{\pi L}_t\geq 0\;(>0)$ holds for all $\pi\in\mathcal{A}_0\;(\pi\in\mathcal{A}_{0,*})$ and  $U(x_0\stochexp{\pi L}_T)>-\infty$, cf.\ \cite[Lemma 2.5]{nutz12bellman}. Here, $x_0\stochexp{\pi L}_t$ is the wealth process corresponding to the constant trading strategy $\pi\in\mathcal{A}$, i.e., it is the solution of
  \begin{equation}\label{eq:sde_wealth}
    X_t=x_0+\int_0^t \pi X_{u-}\;dL_u.
  \end{equation}
  We also consider non-constant admissible trading strategies $\pi_t(\omega)$ given by
  \begin{align*}
      \mathcal{T}&:=\{(\pi_t)_t\rvert\, \pi \text{ pred.\ and } L\text{-integrable, }\pi_t(\omega)\in\calC,\, \stochexp{\pi\sbullet L}_t>0\} \text{ if }p\in(0,1),\\
      \mathcal{T}&:=\{(\pi_t)_t\rvert\, \pi \text{ pred.\ and } L\text{-integrable, }\pi_t(\omega)\in\calC,\, \stochexp{\pi\sbullet L}_t>0\} \text{ if }p\geq1,
  \end{align*}
 
  where $\pi\sbullet L$ denotes the stochastic integral of $\pi_t$ w.r.t.\ $L$.

  It is the agent's aim to find an optimal trading strategy $\pi^*_t$ for the problem of maximizing expected utility problem from terminal wealth on $[0,T]$
  \begin{equation}\label{eq:utility_maximization}
    u(x_0):=\sup_{\pi\in\mathcal{T}}\expval{U(x_0\stochexp{\pi\sbullet L}_T)}=\sup_{\pi\in\mathcal{T}}\expval{\frac{(x_0\stochexp{\pi\sbullet L}_T)^{1-p}}{1-p}}.
  \end{equation}
  In order to discuss the optimal trading strategy $\pi^*_t$ in more detail we now state the assumptions that we make throughout this exposition:
  \begin{assumptions}\label{assumptions}\rule{0pt}{0pt}
    \begin{enumerate}
     \item \label{assump:Spos} $S>0$.     
     \item \label{assump:noarb}$c\neq 0$ or both $F(\BbbR_-)>0$ and $F(\BbbR_+)>0$.
     \item \label{assump:finiteness} $u(x_0)<\infty$, i.e., the maximization problem \eqref{eq:utility_maximization} is finite.
    \end{enumerate}
  \end{assumptions}
  \noindent By \cite[Lemma A.8]{goll_kallsen00} assumption \eqref{assump:Spos} implies $F(-\infty,-1]=0$ and \eqref{assump:Spos} is in fact equivalent to the existence of a \Levy\ process $\widetilde L$ such that $S_t=\stochexp{L}_t=\exp(\widetilde L_t)$. \cite[Lemma A.8]{goll_kallsen00} also gives the following formulae to determine the \Levy\ triplet $(\widetilde b(h), \widetilde c, \widetilde F)$ of $\widetilde L$
    \begin{align}
      \widetilde b(h)&= b(h)-\frac{c}{2}+\int_{(-1,\infty)}\big(h(\log(1+x))-h(x)\big)\; dF(x),\label{eq:widetildeb}\\
      \widetilde c&= c,\label{eq:widetildec}\\
      \widetilde F(G)&=\int 1_G(\log(1+x))\;dF(x)\quad\text{for all Borel sets } G\in\mathcal{B},\label{eq:widetildeF}
    \end{align}
  where $h$ is a bounded truncation function, i.e., $h(x)=0$ for $|x|>R>0$. From now on we \emph{fix} a bounded and continuously differentiable trunction function $h$ and write in the sequel -- if not stated differently -- $b:=b(h)$, $\widetilde b:=\widetilde b(h)$ to alleviate notation.

   As already mentioned in the introduction, by \eqref{assump:noarb} we make the assumption that $L$ has a non-zero Brownian motion part or that $L$ has both positive and negative jumps. Although this is not technically necessary it simplifies further discussion on the admissibility of $N$-period trading strategies below. In particular, assumption \eqref{assump:noarb} assures a no arbitrage condition known as \emph{no unbounded increasing profit}, which has been introduced in \cite{kardaras09no}. We comment on the less important case of a \Levy\ process with $c=0$ and either only positive or only negative jumps at the end of this section.

    Assumption \eqref{assump:finiteness} is a common assumption assuring existence and uniqueness of a maximizer of \eqref{eq:utility_maximization}, i.e., the $\sup$ in \eqref{eq:utility_maximization} is in fact a $\max$.
  \medskip

   Under Assumption \ref{assumptions}, the optimal trading strategy $\pi^*_t$ of \eqref{eq:utility_maximization} becomes trackable and is of a simple structure: Nutz (cf.\ \cite[Theorem 3.2 and Remark 3.3]{nutz09power}, for $p\neq 1$) and Kardaras (cf.\ \cite[Lemma 5.1]{kardaras09no}, for $p=1$) have shown that under our assumptions the optimal trading strategy $\pi^*$ is unique and \emph{constant}, $\pi^*\in\mathcal{A}\cap\calC$. In the following, we denote the optimal strategy of the unconstrained model ($\calC=\BbbR$) by $\pi^*$ and the optimal strategy under the exogenous portfolio constraint $\calC=\C=[0,1]$ by $\pi^*_\C$, which we refer to as the \emph{optimal strategy of the constrained continuous-time model}.

   The optimal wealth processes are given by $x_0\stochexp{\pi^* L}$, $x_0\stochexp{\pi^*_\C L}$ and an investigation of \cite{kramkov_schachermayer99} even shows that the Inada conditions upon $U$ ensure
  \begin{equation}\label{eq:wealth_pos}
    \stochexp{\pi^* L}_t>0 \quad t\in[0,T],
  \end{equation}
   i.e., $\pi^*\in\mathcal{A}_{0,*}$ for all $p>0$. $\stochexp{\pi^*_\C L}_t>0$ holds true since $\mathcal{A}\cap[0,1]=[0,1]$ and $\stochexp{\pi L}_t>0$ for all $\pi\in[0,1]$. \cite[Theorem 3.2]{nutz09power} (for $p\neq1$) and \cite[Lemma 5.1]{KaKa07} (for $p=1$) show that $\pi^*$, $\pi^*_\C$ are given as the unique maximizer of the deterministic concave function $g:\mathcal{A}\cap\calC\to\BbbR$ defined by
  \begin{align}
   \begin{split}\label{eq:g}
    g(\pi)&:=\pi b-\frac{p \pi^2 c}{2}+\int_{(-1,\infty)}\left(\frac{\left(1+\pi x\right)^{1-p}-1}{1-p}-\pi h(x)\right)\; dF(x)\quad \text{for }p\neq 1,\\
    g(\pi)&:=\pi b-\frac{p \pi^2 c}{2}+\int_{(-1,\infty)}\left(\log(1+\pi x)-\pi h(x)\right)\; dF(x)\quad \text{for }p=1,
   \end{split}
  \end{align}
  where we have used $F(-\infty,-1]=0$. In particular, the function $g$ corresponding to the constrained case $\C=[0,1]$ is the restriction on $[0,1]$ of the function $g$ corresponding to unconstrained case $\calC=\BbbR$. This also shows that $\pi^*\in[0,1]$ implies $\pi^*=\pi^*_\C$. Moreover, the concavity of $g$ shows that $\pi^*>1$ ($\pi^*<0$) implies $\pi^*_\C=1$ ($\pi^*_\C=0$).

  \eqref{eq:g} implies that $\pi^*$, $\pi^*_\C$ are independent of the initial capital $x_0$, which is a well known property of power utility functions even in more general market models. By \cite[Theorem 3.2]{nutz09power} the optimal wealth processes $x_0\stochexp{\pi^* L}_t$ and $x_0\stochexp{\pi^*_\C L}_t$, respectively, lead to the terminal utility
  \begin{align}
   \begin{split}\label{eq:term_utility}
    \expval{U(x_0\stochexp{\pi L}_T)}&=\frac{x_0^{1-p}}{1-p}e^{(1-p)g(\pi)T}\quad \text{for }\pi=\pi^*,\pi^*_\C\text{ and }p\neq1,\\
    \expval{U(x_0\stochexp{\pi L}_T)}&=\log(x_0)+g(\pi)T\quad \text{for }\pi=\pi^*,\pi^*_\C\text{ and }p=1.
   \end{split}
  \end{align}
  
  \medskip
  
  We now construct the $N$-period model associated to \eqref{eq:stock_contLevy}: let $\sigma_N:=\{t_0=0,t_1=\frac{T}{N},t_2=\frac{2T}{N},\ldots,t_{N}=T\}$ for $N\in\BbbN$. We define the discrete-time process $S^N$ by
    \begin{equation*}
      S^N_{t_i}:=S_{t_i}\quad i=0,\ldots,N,
    \end{equation*}
  i.e., $S^N$ is the restriction of $S$ to $\sigma_N$. Hence, investing in $S^N$ amounts to investing in $S$ with the restriction that the portfolio may only be adapted at time instants $t_i\in\sigma_N$. From now on we refer to $S^N$ as the \emph{$N$-period model approximation} of \eqref{eq:stock_contLevy}. $S^N$ evolves in each period according to
    \begin{equation*}
      S^N_{t_{i+1}}=S^N_{t_{i}}\left(1+Z^N_i\right)=S_0\prod_{j=0}^{i}\left(1+Z^N_j\right)=:S_0\stochexp{Z^N}_{t_{i+1}},
    \end{equation*}
  where $Z^N_i=\frac{S(t_{i+1})-S(t_i)}{S(t_i)}=\frac{\stochexp{L}_{t_{i+1}}-\stochexp{L}_{t_{i}}}{\stochexp{L}_{t_{i}}}=e^{\widetilde L_{t_{i+1}}-\widetilde L_{t_i}}-1$ are i.i.d.\ random variables. Investing in each period the fraction $\pi$ of current wealth in the stock $S^N$ amounts to the wealth process
    \begin{equation}\label{eq:wealth_N_per}
      x_0\stochexp{\pi Z^N}_{t_{i+1}}:=x_0\prod_{j=0}^i\left(1+\pi Z^N_j\right).
    \end{equation}
  For the same economic agent and utility function $U$ as above, the set of admissible constant trading strategies for $S^N$ is given by
  \begin{align*}
     \mathcal{A}^N:=\{\pi\in\BbbR|1+\pi Z_0^N\geq 0\}=\{\pi\in\BbbR|1+\pi Z_0^N> 0\}=[0,1].     
  \end{align*}
  $\mathcal{A}^N$ is given by $[0,1]$ and does not depend on $p$ and $N$ since we assume that $L$ satisfies $c\neq 0$ or allows both positive and negative jumps. We also find $\mathcal{A}^N=[0,1]\subseteq \mathcal{A}$.

  The set of non-constant admissible trading strategies $(\pi_{i}(\omega))_i$ is given by 
  $$\mathcal{T}^N:=\{(\pi_i)_{i=0,\ldots,N-1}\rvert\, \pi_i \text{ is } \mathcal{F}_{t_i-}\text{ measurable, }\prod_{j=0}^{i}\left(1+\pi_iZ^N_j\right)>0 \}.$$  
  Similarly to \eqref{eq:utility_maximization} it is the agent's aim to find for each $N$ a maximizer $(\pi_{N,i}^*)_{i=0}^{N-1}$ of
    \begin{equation}\label{eq:utility_maximization_N}
      u^N(x_0):=\max_{\pi\in\mathcal{T}^N}\expval{U(x_0\stochexp{\pi Z^N}_T)}=\max_{\pi\in\mathcal{T}^N}\expval{\frac{(x_0\stochexp{\pi Z^N}_T)^{1-p}}{1-p}}.
    \end{equation}

  Maximizing \eqref{eq:utility_maximization_N} leads to a \emph{constant} maximizer $\pi^*_N\in\mathcal{A}^N=[0,1]$, i.e., $\pi_{N,i}^*=\pi_N^*\in[0,1]$. Moreover, $\pi_N^*$ is in fact the maximizer of the following 1-period model
    \begin{equation}\label{eq:utility_maximization_N2}
      \max_{\pi\in[0,1]}\expval{U\left(x_0\left(1+\pi(\stochexp{L}_{\frac{T}{N}}-1)\right)\right)},
    \end{equation}
  and is independent of the initial wealth $x_0$, cf.\ \cite{samuelson69lifetime}. We want to give a short sketch by backward induction why $\pi^*_N$ is constant: Given current wealth $X_{t_{N-1}}=x$, the Bellman principle and $Z^N$ being i.i.d.\ show that the optimal strategy $\pi^*_{N,N-1}(x)$ for the last step maximizes \eqref{eq:utility_maximization_N2} with $x_0=x$. Moreover, since $U$ is a power utility function, $\pi^*_{N,N-1}$ is in fact independent of $x$. Hence, $\pi_{N,N-1}^*=\pi^*_N$. Similar arguments as well as $(Z^N_{i=n+1})^{N}$ being i.i.d.\ -- and in particular independent of the current wealth $X_{t_n}$ -- prove $\pi_{N,n}^*=\pi^*_N$ for the induction step.

  Analogously to $g$ above, we define for each $N$ the deterministic concave function $g^N:[0,1]\to\BbbR$ by
    \begin{align}
      \begin{split}\label{eq:gN}
	 g^N(\pi)&=\frac{N}{T}\expval{\frac{\left(1+\pi (\stochexp{L}_{\frac{T}{N}}-1)\right)^{1-p}-1}{1-p}}\quad\text{for }p\neq 1,\\
	 g^N(\pi)&=\frac{N}{T}\expval{\log\left(1+\pi (\stochexp{L}_{\frac{T}{N}}-1)\right)}\quad\text{for }p= 1.
      \end{split}
    \end{align}
  Clearly, $\pi^*_N$ is a maximizer of $g^N$. The reason why we have slightly modified $g^N$ in comparison to \eqref{eq:utility_maximization_N2} will come apparent in the proof of Lemma \ref{lem:conv_g} below.
  
  \subsection*{The case $c=0$ and either $F(\BbbR_-)=0$ or $F(\BbbR_+)=0$:}
  
  We now shortly discuss the case of $L$ being a pure jump process with triplet $(b,0,F)$ and either $F(\BbbR_-)=0$ or $F(\BbbR_+)=0$. In order to exclude arbitrage, we have to assume that $b$ is negative (positive) if only positive (negative) jumps are allowed.

  The set of admissible constant trading strategies for the $N$-period model is then given by
    \begin{align*}     
     \mathcal{A}^N_{0}&:=\{\pi\in\BbbR|1+\pi Z_0^N=1+\pi(\stochexp{L}_{\frac{1}{N}}-1)\geq 0\}\quad\text{if }p\in(0,1),\\
     \mathcal{A}^N_{0,*}&:=\{\pi\in\BbbR|1+\pi Z_0^N=1+\pi(\stochexp{L}_{\frac{1}{N}}-1)> 0\}\quad\text{if }p\geq 1.
    \end{align*}
  By $S>0$ we have $[0,1]\subseteq\mathcal{A}^N$ and $\mathcal{A}^N$ is bounded for all $N$. We note that $1+\pi(\stochexp{L}_t-1)\geq 0\;(>0)$ implies $1+\pi(\stochexp{L}_s-1)\geq 0\;(>0)$ for all $0\leq s\leq t$, which further leads to $\mathcal{A}^N\subseteq\mathcal{A}^{N+1}$. Letting $t\to0$ in $1+\pi(\stochexp{L}_t-1)\geq 0(>0)$, \cite[II.8a]{jacod_shiryaev03limit} implies $\mathcal{A}^N\subseteq\mathcal{A}$. We see that $\mathcal{A}^N$ depends on $N$ and $[0,1]$ is generally strictly contained in $\mathcal{A}^N$.
  \medskip

  As already mentioned, we make in the sequel the assumption that $c\neq 0$ or both $F(\BbbR_-)>0$ and $F(\BbbR_+)>0$, i.e., $\mathcal{A}^N=[0,1]$ which greatly simplifies further notation. The exogenous portfolio constraint $\C=[0,1]$, which implies the convergence $\lim_{N\to\infty}\pi^*_N=\pi^*_\C$ in the case of $c\neq 0$ or both $F(\BbbR_-)>0$ and $F(\BbbR_+)>0$ (see Theorem \ref{thm:conv_optimal_strategies} below), translates to $\C=\bigcup_{N\geq 1} \overline{\mathcal{A}^N}$ in the case of $c=0$ and either $F(\BbbR_-)=0$ or $F(\BbbR_+)=0$.

\section{Analytic Properties of $g$ and $g^N$}\label{sec:prop_g}
This section summarizes some results on the continuity and differentiability of the functions $g$ and $g^N$, defined in \eqref{eq:g} and \eqref{eq:gN}, respectively.  

\begin{proposition}\label{prop:prop_of_g}
  Under Assumption \ref{assumptions}, $g$ is continuous on $\mathcal{A}_{0,*}$. Moreover, $g$ and $g^N$ are finite and differentiable on $\mathcal{A}_{0,*}$, respectively $[0,1]$ with derivatives
  \begin{align}
    g'(\pi)&=b-p \pi c+\int_{(-1,\infty)}\left(\frac{x}{\left(1+ \pi x\right)^p}- h(x)\right)\; dF(x),\label{eq:diff_g}\\ 
    \left(g^N\right)'(\pi)&=\frac{N}{T}\expval{\frac{\exp(\widetilde L_\frac{T}{N})-1}{\left(1+\pi (\exp(\widetilde L_\frac{T}{N})-1)\right)^{p}}},\label{eq:diff_gN}
  \end{align}
  that are finite on $\left(\mathcal{A}_{0,*}\right)^\circ$, respectively $(0,1)$, where $\left(\mathcal{A}_{0,*}\right)^\circ$ denotes the interior of  $\mathcal{A}_{0,*}$
\end{proposition}

 For the proof of Proposition \ref{prop:prop_of_g} we need the following elementary result on exchanging differentiation and integration for concave functions, which is  a special case of \cite[Lemma 5.14]{nutz12bellman}.
\begin{lemma}\label{lem:interchange_diff_int}
    Let $\eta$ be a Borel-measure on $\BbbR$, $y_0\in\BbbR$ and let $f:\BbbR\times\BbbR\to\BbbR$ satisfy
    \begin{enumerate}
     \item \label{item:interchange_diff_int_1}$x\mapsto f(x,y)$ is measurable and $\int_\BbbR |f(x,y)|d\eta(x)<\infty$ for $y$ in a neighbourhood of $y_0$, and
     \item \label{item:interchange_diff_int_2} $\forall x\in\BbbR:$ $y\mapsto f(x,y)$ is concave and differentiable.      
    \end{enumerate}
    Then $k(y):=\int_\BbbR f(x,y)d\eta(x)$ is concave and differentiable with derivative
    \begin{equation*}
      k'(y_0)=\int_\BbbR f'(x,y_0)d\eta(x),
    \end{equation*}
    where $f'(x,y_0)$ denotes the partial derivative of $f$ w.r.t.\ the second argument $y$ at $(x,y_0)$. 
  \end{lemma}
  \begin{proof} Cf.\ \cite[Lemma 5.14]{nutz12bellman}.
  \end{proof}

\begin{proof}[Proof of Proposition \ref{prop:prop_of_g}]
  Since $u(x_0)<\infty$ we find in particular that $g^N$ is finite on $[0,1]$. By \cite[Corollary 3.7]{nutz09power}, $u(x_0)<\infty$ is equivalent to 
   \begin{equation}\label{eq:prop_of_g1}
      \int_{x>1}\frac{(1+x)^{1-p}}{1-p}\;dF(x)<\infty \quad\text{for }p>0,
   \end{equation}
  and \cite[Lemma 5.3]{nutz09power} gives that $g$ is finite and continuous on $\mathcal{A}_{0,*}$. 

  Lemma \ref{lem:interchange_diff_int} implies that $g$ and $g^N$ are (left- and right-) differentiable on $\mathcal{A}_{0,*}$ and $[0,1]$, respectively, and their derivatives are given by \eqref{eq:diff_g} and \eqref{eq:diff_gN}, respectively. The concavity of $g$ and $g^N$ implies $|g'|,|(g^N)'|<\infty$ on $\left(\mathcal{A}_{0,*}\right)^\circ$ and $(0,1)$, respectively.
\end{proof}

\section{Convergence Results for the Exponential \Levy\ model}\label{sec:conv_g}
  Our main convergence result is stated in Theorem \ref{thm:conv_optimal_strategies} below, which also implies two important corollaries. Proposition \ref{prop:levy_order_of_strategies} gives qualitative properties of the optimal strategies $\pi_N^*$ and $\pi^*$ related to their sign and the risk aversion of $U$. To improve readability, we move the proofs of Theorem \ref{thm:conv_optimal_strategies} and Proposition \ref{prop:levy_order_of_strategies} to the end of this section. We comment on the extension of Theorem \ref{thm:conv_optimal_strategies} to higher dimensions at the end of its proof.
  \begin{theorem}\label{thm:conv_optimal_strategies}
     Under Assumption \ref{assumptions}, $g^N$ converges uniformly on $[0,1]$ to $g$. In particular, $\lim_{N\to\infty}\pi^*_N=\pi^*_\C$ and $\lim_{N\to\infty}g^N(\pi^*_{N})=g(\pi^*_\C)$. Thus, if $\pi^*\in[0,1]$ then $\lim_{N\to\infty}\pi^*_N=\pi^*$, and $\lim_{N\to\infty}g^N(\pi^*_{N})=g(\pi^*)$.
  \end{theorem}  

\begin{corollary}\label{cor:conv_utility}
  Under Assumption \ref{assumptions}, the $N$-period value functions converge to the value function of the constrained continuous-time model, i.e.,
  \begin{equation*}
    \lim_{N\to\infty}\expval{U\left(x_0\stochexp{\pi^*_N Z^N}_T\right)}=\expval{U(x_0\stochexp{\pi^*_\C L}_T)}.
  \end{equation*}
  Thus, if $\pi^*\in[0,1]$ then
  \begin{equation*}
    \lim_{N\to\infty}\expval{U\left(x_0\stochexp{\pi^*_N Z^N}_T\right)}=\expval{U(x_0\stochexp{\pi^* L}_T)}.
  \end{equation*}
\end{corollary}
\begin{proof}
  We give the proof for $0<p$, $p\neq1$, since the case $p=1$ is shown similarly. W.l.o.g.\ we set $x_0=1$. By \eqref{eq:wealth_N_per} and $Z_j$ being i.i.d.\ we find
   \begin{equation*}
      \expval{U\left(\stochexp{\pi^*_N Z^N}_T\right)}=\frac{\expval{(1+\pi_N^*Z^N_0)^{1-p}}^{N}}{1-p}=\frac{\left(1+\frac{(1-p)Tg^N(\pi_N^*)}{N}\right)^{N}}{1-p}.
   \end{equation*}
  Since $(1+\frac{a}{N})^{N}\to\exp(a)$ converges uniformly on compacts and $g^N(\pi_N^*)\to g(\pi^*_\C)$ as $N\to\infty$, equation \eqref{eq:term_utility} shows
    \begin{equation*}
     \expval{U\left(\stochexp{\pi^*_N Z^N}_T\right)}\to\frac{\exp\left((1-p)Tg(\pi^*_\C)\right)}{1-p}=\expval{U(\stochexp{\pi^*_\C L}_T)}.
    \end{equation*}
\end{proof}	

  The proof of the next corollary of Theorem \ref{thm:conv_optimal_strategies} requires some results on the convergence of Euler approximations of \eqref{eq:sde_wealth}, which can be found in the Appendix.

  \begin{corollary}\label{cor:conv_wealth}
    Let Assumption \ref{assumptions} be satisfied and let $L$ be square integrable. Then the optimal terminal payoffs of the $N$-period models converge in $L^2(\Omega)$ to the optimal terminal payoff of the constrained continuous-time model, i.e.,
    \begin{equation*}
      \lim_{N\to\infty}\expval{\left(x_0\stochexp{\pi^*_N Z^N}_T-x_0\stochexp{\pi^*_\C L}_T\right)^2}=0.
    \end{equation*}  
    Thus, if $\pi^*\in[0,1]$ then
    \begin{equation*}
      \lim_{N\to\infty}\expval{\left(x_0\stochexp{\pi^*_N Z^N}_T-x_0\stochexp{\pi^* L}_T\right)^2}=0.
    \end{equation*} 
  \end{corollary}
  \begin{proof}
     W.l.o.g.\ we set $x_0=1$. By Proposition \ref{prop:conv_stoch_euler} in the Appendix $(\stochexp{\pi^*_N Z^N}_T-\stochexp{\pi^*_N L}_T)$ converges in $L^2(\Omega)$ to 0, where $\stochexp{\pi^*_N L}$ denotes the Euler approximation of $\stochexp{\pi^*_N L}$, cf.\ \eqref{eq:euler_approx} in the Appendix.
   
    $\pi^*_N\in[0,1]$ implies $\stochexp{\pi^*_N L}>0$. In particular, Theorem \ref{thm:conv_euler_approx} in the Appendix implies that the Euler approximations $\stochexp{\pi^*_N L}_T$ converge to $\stochexp{\pi^*_\C L}_T$ as $N\to\infty$. In total, $\stochexp{\pi^*_N Z^N}_T$ converges in $L^2(\Omega)$ to $\stochexp{\pi^*_\C L}_T$.
  \end{proof}

  \begin{proposition}\label{prop:levy_order_of_strategies}    
    Let Assumption \ref{assumptions} be satisfied and let $L$ be integrable, i.e., we can choose $h(x)=x$.
    \begin{enumerate}
     \item \label{it:sign} Letting $b=b(x)$ denote the drift w.r.t.\ $h(x)=x$, we find
      \begin{align}\label{eq:lem_sgn_strategy_result}
	  b\geq 0 \Longrightarrow \pi^*,\pi^*_N\geq 0 \quad\text{and}\quad b\leq 0 \Longrightarrow \pi^*\leq 0, \pi^*_N=0.
      \end{align}
      If $\pi^*\geq 1$, then $\pi^*_N=1$.
    \item \label{it:prop_mon} Assume that the maximum in \eqref{eq:utility_maximization} is finite for all $U=\frac{x^{1-p}}{1-p}$ with $p\in(p_a,p_b)$. We denote by $\pi_p^*$, $\pi_{p,N}^*$ the corresponding optimal strategy for $U$ in the continuous-time and discrete-time model, respectively. Then the mappings $p\mapsto\pi_p^*$, $p\mapsto\pi_{p,N}^*$ are montone decreasing $(\text{increasing})$ on $(p_a,p_b)$, if $b>0$ $(b<0)$, where $b=b(x)$ is the drift w.r.t.\ $h(x)=x$.
  \end{enumerate}     
  \end{proposition}  
\begin{remark}
  \begin{enumerate}
    \item \eqref{eq:lem_sgn_strategy_result} can be shown in a more general framework using sub- and supermartingale arguments. However, we want to establish \eqref{it:sign} by taking advantage of the special structure of the exponential \Levy\ model and prove it by montonicity properities of the functions $g$ and $g^N$, see the proof below.
    \item By $-\frac{U''(x)}{U'(x)}=\frac{p}{x}$ the constant $p$ determines the index of relative risk aversion of $U$. Proposition \ref{prop:levy_order_of_strategies}.\eqref{it:prop_mon} formalizes the intuitive idea that higher risk aversion leads to a lower fraction of wealth invested in the stock $S$.  
  \end{enumerate}  
\end{remark}

\subsection*{Proof of Theorem \ref{thm:conv_optimal_strategies}:}
  The proof of Theorem \ref{thm:conv_optimal_strategies} requires the pointwise convergence of $g^N$ to $g$ that is established by Lemma \ref{lem:conv_g}.
  \begin{lemma}\label{lem:conv_g}
    Let Assumption \ref{assumptions} be satisfied and let $\pi\in[0,1]$. Then
      \begin{equation*}
	\lim_{N\to\infty}g^N(\pi)= g(\pi).
      \end{equation*}
  \end{lemma}
  \begin{proof}
    We only give the proof for $0<p$, $p\neq 1$, since the case $p=1$ is obtained similarly. A Taylor expansion of $x\mapsto \frac{\left(1+\pi \left(e^x-1\right)\right)^{1-p}}{1-p}$ around $0$ shows
    \begin{equation*}
      \frac{2}{\pi-p\pi^2}\left(\frac{\left(1+\pi (e^x-1)\right)^{1-p}-1}{1-p}-\pi x\right)\sim x^2\quad \text{as } x\to 0,
    \end{equation*}
    where we use the common notation $f(x)\sim g(x),\; x\to 0\Leftrightarrow \lim_{x\to 0}\frac{f(x)}{g(x)}=1$. In particular, we find 
    \begin{equation*}
      k(x):=\frac{2}{\pi-p\pi^2}\left(\frac{\left(1+\pi (e^x-1)\right)^{1-p}-1}{1-p}-\pi h(x)\right)\sim x^2\quad \text{as } x\to 0.
    \end{equation*}    
    Since $k$ is locally bounded we can apply \cite[Theorem 1.1 (ii)]{figueroa08small}, which implies
    \begin{equation}\label{eq:lem_conv_g}
      \lim_{N\to \infty}\frac{N}{T}\expval{k(\widetilde L_\frac{T}{N})}=\widetilde{c}+\int_\BbbR k(x)\;d\widetilde F(x).
    \end{equation}
    \cite[Lemma 4.2]{jacod07asymptotics} implies $\lim_{N\to\infty}\frac{N}{T} \pi h(\widetilde L_\frac{T}{N})=\pi\widetilde b(h)$. Shifting this limit on the right side of \eqref{eq:lem_conv_g} and using equations \eqref{eq:widetildeb}-\eqref{eq:widetildeF} we finally get
    \begin{equation*}
      \lim_{N\to \infty} g^N(\pi)=\pi b-\frac{p \pi^2 c}{2}+\int_{(-1,\infty)}\left(\frac{\left(1+\pi x\right)^{1-p}-1}{1-p}-\pi h(x)\right)\;dF(x).
    \end{equation*}
  \end{proof}

  \begin{proof}[Proof of Theorem \ref{thm:conv_optimal_strategies}]
     By Lemma \ref{lem:conv_g} we know that $g^N$ converges pointwise to $g$ on $[0,1]$. Lemma \ref{lem:pointwise_uniform} in the Appendix shows that   the concavity of $g^N$ and $g$ imply that $g^N$ converges in fact \emph{uniformly} to $g$ on $[0,1]$ as $N\to\infty$. Thus, the $N$-period maximizers $\pi_N^*$ satisfy $\lim_{N\to\infty}\pi_N^*=\arg\max_{[0,1]}g=\pi^*_\C$ and $\lim_{N\to\infty}g^N(\pi^*_N)= g(\pi^*_\C)$. 
  \end{proof} 

  \subsubsection*{Extensions to Higher Dimensions}
  If the \Levy\ process $L=(L_1,\ldots,L_n)$ and the stock $S=\stochexp{L}=(\stochexp{L_1},\ldots,\stochexp{L_n})$ are $\BbbR^n$ valued, the only difficulty in extending our results to higher dimensions is the no-arbitrage condition of \emph{no unbounded increasing profit}, cf.\ \cite{kardaras09no}. As we have seen in Assumption \ref{assumptions}.\eqref{assump:noarb} this condition is almost trivially satisfied in one dimension, but it gets more involved in higher dimensions. Lemma \ref{lem:conv_g} also proves the pointwise convergence of $g^N$ to $g$ on $[0,1]^n$, where $g$ and $g^N$ are defined verbatim in higher dimensions. This finally proves by concavity the uniform convergence of $g^N$ to $g$ on $[0,1]^n$. Thus, under appropriate no-arbitrage assumptions Theorem \ref{thm:conv_optimal_strategies} also holds in higher dimensions.

  \subsection*{Proof of Proposition \ref{prop:levy_order_of_strategies}:}
  \begin{proof} \emph{ad \eqref{it:sign}}: Using \eqref{eq:diff_g} with $h(x)=x$ we get $g'(0)=b=b(x)$, which immediately gives \eqref{eq:lem_sgn_strategy_result}. Moreover, \eqref{eq:diff_gN} implies $(g^N)'(0)=NT^{-1}(e^{bN^{-1}T}-1)$. Hence, $b\leq 0$ implies $\pi^*_N=0$ and $b\geq 0$ gives $\pi^*_N\geq 0$. If $\pi^*\geq 1$ we find in particular $0\leq g'(1)<\infty$. Using \eqref{eq:widetildeb}-\eqref{eq:widetildeF} a straight forward calculation shows
    \begin{equation*}
      (g^N)'(1)=\exp\left(-p\widetilde b+\frac{p^2c}{2}+\int_\BbbR \big(e^{-px}-1+ph(x)\big)\,d\widetilde F\right)\cdot(e^{g'(1)}-1)\geq 0,
    \end{equation*}
    thus $\pi_N^*=1$.

    \noindent\emph{ad \eqref{it:prop_mon}}: Let $b\neq 0$ since otherwise $\pi^*=\pi^*_N=0$. We first consider the continuous-time case: letting $p_1<p_2$, $p_1,p_2\in(p_a,p_b)$, \eqref{it:sign} shows that we may assume $|\pi_{p_1}^*|,|\pi_{p_2}^*|>0$ and that $\pi_{p_1}^*$, $\pi_{p_2}^*$ have the same sign. Hence, in order to prove \eqref{it:prop_mon} we have to establish $|\pi_{p_1}^*|\geq |\pi_{p_2}^*|$.

    We first consider the case when $\pi_{p_1}^*$ and $\pi_{p_2}^*$ are in the interior $(\mathcal{A}_{0,*})^\circ$ of $\mathcal{A}_{0,*}$. As $\pi_{p_i}^*\neq0$, $\pi_{p_i}^*$ satisfy in particular the equation
    \begin{equation*}
      G(\pi_{p_i}^*,p_i):=\pi_{p_i}^* b-p_i(\pi_{p_i}^*)^2 c+\int_{(-1,\infty)}\left(\frac{\pi_{p_i}^* x}{(1+\pi_{p_i}^* x)^{p_i}}-\pi_{p_i}^*x\right)\,dF(x)=0 
    \end{equation*}
    with $i=1,2$, since $G(\pi,p)=\pi g'(\pi,p)$. Here $g'(\cdot,p)$ denotes the partial derivative w.r.t.\ $\pi$ of the function $g$ corresponding to $p$.

      The integral $I(\pi,p):=\int_{(-1,\infty)}x(1+\pi x)^{-p}-x\,dF(x)$ is partially differentiable w.r.t.\ $\pi$ and $p$: indeed, the integrand $(\pi,p)\mapsto x(1+\pi x)^{-p}-x$ is convex (concave) in $\pi$ (with $p$ fixed) and in $p$ (with $\pi$ fixed) if $x\geq 0$ ($x\leq 0$). Hence, we can apply Lemma \ref{lem:interchange_diff_int} by splitting the domain of integration into $(-1,0]$ and $[0,\infty)$, which yields that $I$ is partially differentiable w.r.t.\ $\pi\in(\mathcal{A}_{0,*})^\circ\backslash\{0\}$ and $p\in(p_a,p_b)$. Moreover, differentiation and integration can be interchanged. In particular, the partial derivatives of $G$ are finite for $\pi\in(\mathcal{A}_{0,*})^\circ\backslash\{0\}$ and $p\in(p_a,p_b)$. Lemma \ref{lem:cont_integral} in the Appendix shows that the partial derivatives of $G$ are continuous for any $p\in(p_a,p_b)$ and $\pi\in(\mathcal{A}_{0,*})^\circ\backslash\{0\}$, which further implies that $G$ is continuously differentiable.
 
    Hence, any pair of optimal strategies $(\pi_p^*,p)$ with $G(\pi_p^*,p)=0$ satisfies
    \begin{align*}      
      \pi_p^*\frac{\partial G}{\partial\pi}(\pi_p^*,p)&=G(\pi_p^*,p)-p(\pi_p^*)^2c-\int_{(-1,\infty)}\pi^2 x^2p(1+\pi_p^* x)^{-p-1}\;dF(x),\\
      \Longleftrightarrow\;\frac{\partial G}{\partial\pi}(\pi_p^*,p)&=-p\pi_p^*\left(c+\int_{(-1,\infty)} x^2(1+\pi_p^* x)^{-p-1}\;dF(x)\right),
    \end{align*}
    where we have used $|\pi_p^*|>0$. In particular, $\text{sgn}\left(\frac{\partial G}{\partial \pi}(\pi_p^*,p)\right)=-\text{sgn}(\pi_p^*)\neq 0$.

    By the implicit function theorem, there exists a differentiable function $\phi$ defined on a neighbourhood $N$ of $p_1$ such that $\forall p\in N:G(\phi(p),p)=0$, $\phi(p_1)=\pi^*_{p_1}$ and 
    \begin{equation}\label{eq:prop_levy_order_of_strategies}
      \phi'(p)=-\left(\frac{\partial G}{\partial \pi}(\phi(p),p)\right)^{-1}\frac{\partial G}{\partial p}(\phi(p),p).
    \end{equation}
    By continuity we may also assume that $\phi(p)\neq 0$ for all $p\in N$. Differentiating $G$ with respect to $p$ implies for all $\pi\in\mathcal{A}_{0,*}\backslash\{0\}$   
    \begin{equation*}\label{eq:levy_order_of_strategies_dG_dp}
      \frac{\partial G}{\partial p}(\pi,p)=-\pi^2c-\int_{(-1,\infty)}\frac{\pi x\log(1+\pi x)}{\left(1+\pi x\right)^p}\;dF(x)<0,
    \end{equation*}
    where the right inequality follows from $(1+\pi x)>0$ $F$-a.s.\ and $\pi x\log(1+\pi x)\geq 0$. Hence, \eqref{eq:prop_levy_order_of_strategies} yields $\text{sgn}(\phi'(p))=-\text{sgn}(\phi(p))$. This also implies that we may assume w.l.o.g.\ that $\phi$ is defined on an interval containing $[p_1,p_2]$: let $p\geq p_1$ be at the right boundary of $N$ and $\pi_p:=\lim_{N\ni p_n\to p}\phi(p_n)$, then $|\phi(p_1)|=|\pi_{p_1}^*|\geq |\phi(p)|=|\pi_{p}|$ and $$G(\pi_p,p):=\lim_{N\ni p_n\to p}G(\phi(p_n),p_n)=\lim_{N\ni p_n\to p}0=0.$$ The continuity of $G(\pi,p)=\pi g'(\pi,p)$ as well as  the uniqueness of a solution $\pi^*_p\neq 0$ of the utility maximization problem related to $p$ implies $g'(\pi_p,p)=0$, $\pi^*_p=\pi_p\neq 0$ and $G(\pi_p^*,p)=0$. By $|\frac{\partial G}{\partial \pi}(\pi_p^*,p)|\neq 0$ the implicit function theorem can also be applied at $p$, which implies that $\phi$ can be defined on an interval containing $[p_1,p_2]$. Finally, $\text{sgn}(\phi'(p))= -\text{sgn}(\phi(p))$ shows $|\phi(p_1)|=|\pi_{p_1}^*|\geq |\pi_{p_2}^*|=|\phi(p_2)|$, as desired.
    \medskip

    If $\pi_{p_2}^*$ is at the boundary of $\mathcal{A}_{0,*}$, we only have to treat the case when $\pi_{p_2}^*$ is at the right boundary while $\pi_{p_1}^*<\pi_{p_2}^*$. We then find
    \begin{equation*}
	G(\pi,p_2)\geq0\quad\forall\pi\in\mathcal{A}_{0,*}\quad\text{and}\quad G(\pi_{p_1}^*,p_1)\leq0.
    \end{equation*}
    Since $\frac{\partial G}{\partial p}(\pi,p)<0$ we find $G(\pi_{p_1}^*,p_2)< 0$, which contradicts the above assumption. Thus $|\pi_{p_1}^*|\geq |\pi_{p_2}^*|$. The case of $\pi_{p_1}^*$ being at the boundary of $\mathcal{A}_{0,*}$ is treated analogously.
   \medskip
    
    In the discrete-time case, similar arguments as above with $G(\pi,p):=\pi\mathbb{E}[(1+\pi(\exp(\widetilde L_{\frac{1}{N}})-1))^{-p}(\exp(\widetilde L_{\frac{1}{N}})-1)]$ show that \eqref{it:prop_mon} also holds for $\pi^*_N$. 
  \end{proof}


  \renewcommand{\theequation}{A.\arabic{equation}}
\appendix
  \renewcommand{\thesection}{\Alph{section}}
  \renewcommand{\thedefinition}{\Alph{section}.\arabic{definition}}
\section{}


The Appendix contains results on the integrability of the stochastic exponential (Lemma \ref{lem:stochexp_moment}), on the pointwise and uniform convergence of concave functions (Lemma \ref{lem:pointwise_uniform}) and on the convergence of the optimal $N$-period wealth processes $\stochexp{\pi_N^* Z^N}$ to the Euler approximation of $\stochexp{\pi^*_\C L}$, i.e., the optimal wealth process of the constrained continuous-time model (Proposition \ref{prop:conv_stoch_euler}). In the proof of Corollary \ref{cor:conv_wealth} we also need the convergence of the Euler approximations to $\stochexp{\pi^*_\C L}$. More explicitly, for our purpose we need that the convergence holds in fact \emph{uniformly} in $\pi^*_\C$ (Theorem \ref{thm:conv_euler_approx}). Since this result only requires a slight modification of a result of Kohatsu and Protter (\cite{kohatsu_protter94euler}), we omit a detailed proof of Theorem \ref{thm:conv_euler_approx} and just remark on the adaption of the proof in \cite{kohatsu_protter94euler}. At the end of the Appendix we also present a technical continuity result (Lemma \ref{lem:cont_integral}) that we have omitted in the proof of Proposition \ref{prop:levy_order_of_strategies}.

\begin{lemma}\label{lem:stochexp_moment}
    Let $L$ be a \Levy\ process satisfying $\expval{(L_t)^p}<\infty$ with $p>0$ and $\stochexp{L}_t>0$. Then there exists a constant $c(p)\in\BbbR$ such that
      \begin{equation}\label{eq:stochexp_moment}
	\expval{\stochexp{L}_t^p}=e^{c(p)t}.
      \end{equation}    
  \end{lemma}
  \begin{proof}
    Using $\stochexp{L}_t=\exp(\widetilde L_t)$ this is a simple consequence of \cite[Theorem 25.3]{sato99levy}.
  \end{proof}

\begin{lemma}\label{lem:pointwise_uniform}
  Let $c,c_n:[0,1]\to\BbbR$ be concave functions and let $c$ be continuous on $[0,1]$. Then the pointwise convergence of $c_n$ to $c$ on $[0,1]$ implies the uniform convergence of $c_n$ to $c$ on $[0,1]$.
\end{lemma}
\begin{proof}
 Let $\sigma_N:=\{t_i=\frac{i}{N}|i=0,\ldots,N\}$ and let $\epsilon>0$ be given. Define for any function $f:[0,1]\to\BbbR$
  $$ \eta(f,N):=\max_{i=0,\ldots,N}\left|f\left(t_i\right)-f\left(t_{i+1}\right)\right|.$$
 Define $[x]_N$ for all $x\in[0,1]$ by $[x]_N:=t_i$ if $t_i\leq x<t_{i+1}$. Any concave function $\tilde c$ satisfies
  $$|\tilde c(x)-\tilde c([x]_N)|\leq 2 \eta(\tilde c,N)\quad \text{ for }N\geq 2.$$
  By assumption there exists $M(N)\in\BbbN$ such that $\max_{i=0,\ldots,N}|c_n(t_i)-c(t_i)|<\epsilon$ for all $n\geq M(N)$. Hence, we find for all $n\geq M(N)$ and all $x\in[0,1]$
  $$|c(x)-c_n(x)|\leq 2\eta(c,N)+2\eta(c_n,N)+\left|c_n([x]_N)-c([x]_N)\right|.$$
  The uniform continuity of $c$ implies that we can choose $N$ large enough such that $\eta(c,N)<\epsilon$. Together with the last inequality this further implies $\max_{x\in[0,1]}|c(x)-c_n(x)|<9\epsilon$ for all $n\geq M(N)$.
\end{proof}

  \begin{theorem}\label{thm:conv_euler_approx}
     Let $L$ be square integrable and $X^N$ be positive and solution of
    \begin{equation*}
      X_t^{N}=X_0+\int_0^t\pi_NX_{s-}^N\;dL_s,
    \end{equation*}
    where $\pi_N$ is a fixed sequence of real numbers converging to $\pi$, i.e., $X^N=\stochexp{\pi_N L}$. Let $\sigma_n:=\{t_0=0,t_1=\frac{T}{n},\ldots,t_{n}=T\}$ and define by $X^{N,n}_{t_k}=\prod_{i=0}^{k-1}(1+\pi_N\Delta_i^nL)=:\stochexp{\pi_N \Delta^nL}_{t_k}$ the Euler approximation of $X^N$, where $\Delta_i^nL:=L_{\frac{(i+1)T}{n}}-L_\frac{iT}{n}$. Then  
    \begin{equation*}
       \lim_{n\to\infty}\expval{\sup_{t\leq T}|X^{N,n}_t-X^N_t|^2}=0,
    \end{equation*}
      and the limit holds \emph{uniformly} in $N$. In particular, under the above assumptions $\lim_{N\to\infty}\expval{(\stochexp{\pi_N \Delta^N L}_T-\stochexp{\pi L}_T)^2}$.
  \end{theorem}
  \begin{proof}
    This result follows in fact from a slight modification of \cite[Theorem 2.5]{kohatsu_protter94euler}. Kohatsu and Protter prove the result for SDEs of the type $X_t^N=X_0+\int_0^tF^N(X)_{s-}\,dY_s$, where $F$ is \emph{bounded} and $Y$ is a special semimartingale. Although in our case $F^N(X)_s=\pi_NX_s$ is \emph{unbounded}, a moment's reflection reveals that the independence of the increments of $Y=L$ implies that one only has to assume $\expval{(X^{N,n}_{t_k})^2}<C<\infty$ with $C$ independent of $k$, which is satisfied since $L$ is square integrable. As $\pi_N\to\pi$ it is also not hard to see that the above convergence holds uniformly in $N$.
  \end{proof}
  To simplify the reading of Proposition \ref{prop:conv_stoch_euler}, we recall the following notation 
  \begin{align}
    \stochexp{\pi_N Z^n}_{t_{k}}&=\prod_{j=0}^{k-1}\left(1+\pi_NZ^n_j\right) \quad\text{ where } \quad Z^n_j=\exp(\widetilde L_{t_{j+1}}-\widetilde L_{t_j})-1,\notag\\
    \stochexp{\pi_N \Delta^nL}_{t_k}&=\prod_{j=0}^{k-1}(1+\pi_N\Delta_j^nL) \quad\text{ where } \quad\Delta_j^nL=L_{t_{j+1}}-L_{t_j},\label{eq:euler_approx}
  \end{align}
  with $t_i=\frac{iT}{n}$.
  \begin{proposition}\label{prop:conv_stoch_euler}
  Let $\pi_N$ be a sequence of real numbers converging to $\pi$ and let $L$ be square integrable. Then
  \begin{equation*}
    \lim_{n\to\infty}\expval{\left(\stochexp{\pi_N Z^n}_T-\stochexp{\pi_N \Delta^nL}_T\right)^2}=0,
  \end{equation*}
  where the limit holds \emph{uniformly} in $N$. In particular, under the above assumptions $\lim_{N\to\infty}\expval{(\stochexp{\pi_N Z^N}_T-\stochexp{\pi_N \Delta^NL}_T)^2}=0$.
\end{proposition}
\begin{proof}
  Letting $X_{t_k}^n:=\stochexp{\pi_N Z^n}_{t_k}$ and $Y_{t_k}:=\stochexp{\pi_N \Delta^nL}_{t_k}$, they satisfy
  \begin{align}\label{eq:prop_conv_stoch_euler}
    X_{t_{k+1}}^{n}=X_{t_k}^{n}+\pi_NX_{t_k}^{n}Z^n_k\quad\text{and} \quad Y_{t_{k+1}}^{n}=Y_{t_k}^{n}+\pi_NY_{t_k}^{n}\Delta^n_k L.
  \end{align}
  Defining the i.i.d.\ random variables $\epsilon^n_k:=Z^n_k-\Delta^n_k L=\frac{\stochexp{L}_\frac{k+1}{n}-\stochexp{L}_\frac{k}{n}}{\stochexp{L}_\frac{k}{n}}-\left(L_\frac{k+1}{n}-L_\frac{k}{n}\right)$, we deduce from \eqref{eq:prop_conv_stoch_euler}
  \begin{align*}
    \expval{\left(X_{t_{k+1}}^{n}-Y_{t_{k+1}}^{n}\right)^2}\leq 3\expval{\left(X_{t_{k}}^{n}-Y_{t_{k}}^{n}\right)^2}+ 3\pi_N^2&\expval{\left(X_{t_{k}}^{n}-Y_{t_{k}}^{n}\right)^2\left(\Delta^n_k L\right)^2}\\
	&+3\pi_N^2\expval{\left(X_{t_k}^n\epsilon^n_k\right)^2}.
  \end{align*}
  Using the independence of $\epsilon^n_k$ and $\mathcal{F}_{t_k}$ as well as of $\Delta^n_k L$ and $\mathcal{F}_{t_k}$, and the estimate $\expval{\left(\Delta^n_k L\right)^2}\leq\frac{C_1}{n}$ -- where $C_1>0$ only depends on the \Levy\ triplet of $L$ -- this further implies
   \begin{equation*}
      \expval{\left(X_{t_{k+1}}^{n}-Y_{t_{k+1}}^{n}\right)^2}\leq 3\frac{n+\pi_N^2C_1}{n}\expval{\left(X_{t_{k}}^{n}-Y_{t_{k}}^{n}\right)^2}+3\pi_N^2\expval{\left(X_{t_k}^n\right)^2}\expval{\left(\epsilon_k^n\right)^2}.
   \end{equation*}
  Since $L$ is square integrable, Lemma \ref{lem:stochexp_moment} implies that $Z^n_k$ is also square integrable. This further gives $\expval{(X^n_{t_k})^2}\leq C_2$, where $C_2$ can be chosen to be independent of $n$ and $N$. The discrete Gronwall inequality of Lemma \ref{lem:gronwall} as well as $X_0^n=Y_0^n$ imply 
  \begin{equation*}
    \expval{\left(X_T^{n}-Y_T^{n}\right)^2}\leq C_3\sum_{i=0}^{n-1}\expval{\left(\epsilon^n_i\right)^2}e^{\sum_{i=0}^{n-1}\frac{\pi_N^2 C_1}{n}}\leq C_4 n\expval{\left(\epsilon^n_0\right)^2},
  \end{equation*}
  where $C_4$ is independent of $n$ and $N$, since $(\pi_N)_N$ are bounded.

  We now establish $\mathbb{E}[\left(\epsilon^n_0\right)^2]=O(n^{-2})$ which concludes the proof: since $L$ is square integrable with triplet $(b(x),c,F)$, we can decompose it canonically into $L_t=(L_t-bt)+bt$, where $M_t:=(L_t-b t)$ is a square integrable martingale and $V_t:=b t$ is the predictable finite variation process. Denoting by $\langle M,M\rangle$ the predictable quadratic variation of $M$, we find $\langle M,M\rangle_t=(c+\int x^2\; F(dx))t$. Lemma \ref{lem:stochexp_moment} and \cite[Theorem 4.40]{jacod_shiryaev03limit} show that $\int_0^{1/n}\stochexp{L}_{s-}-1\,dM_s$ is a square integrable martingale. Using the It\^{o} isometry (for $M_t$) as well as Jensen's inequality (for $V_t$), we find 
    \begin{equation*}
	  \expval{\left(\epsilon^n_0\right)^2}=\expval{\left(\int_0^\frac{1}{n}\stochexp{L}_{s-}-1dL_s\right)^2}\leq \frac{C_5n+C_5}{n}\int_0^\frac{1}{n}\expval{(\stochexp{L}_{s}-1)^2}ds,
    \end{equation*}
  where $C_5$ only depends on the \Levy\ triplet of $L$. Using the same arguments as before we find for $0\leq s\leq \frac{1}{n}$
    \begin{align*}
      \expval{(\stochexp{L}_{s}-1)^2}&=\expval{\left(\int_0^{s}\stochexp{L}_{u-}\,dL_u\right)^2},\\
	  &\leq \frac{C_5n+C_5}{n}\int_0^{s}\expval{(\stochexp{L}_{u})^2}\,du\leq \frac{C_6}{n},
    \end{align*}
  where the last inequality follows from $L$ being square integrable and Lemma \ref{lem:stochexp_moment}. Hence, $\expval{\left(\epsilon^n_0\right)^2}=O(n^{-2})$.  
\end{proof}

\begin{lemma}[Discrete Gronwall Lemma]\label{lem:gronwall}
  Let $(\delta_i)_{i=0}^N$, $(e_i)_{i=0}^N$, $(\eta_i)_{i=0}^N$ satisfy $\delta_i,e_i,\eta_i\geq 0$ and 
    \begin{equation*}
      e_{i+1}\leq (1+\delta_i)e_i+\eta_i\quad i=0,\ldots, N-1.
    \end{equation*}
  Then
    \begin{equation*}
      e_i\leq\left(e_0+\sum_{j=0}^{i-1}\eta_j\right)\exp\big(\sum_{j=0}^{i-1}\delta_j\big) \quad i=0,\ldots,N.
    \end{equation*}
\end{lemma}
\begin{proof}
   Straight forward using induction.
\end{proof}

\begin{lemma}[omitted part of the proof of Proposition \ref{prop:levy_order_of_strategies}.\eqref{it:prop_mon}]\label{lem:cont_integral}
  Under the assumptions of Proposition \ref{prop:levy_order_of_strategies}.\eqref{it:prop_mon}, the functions
  \begin{align*}
    K(\pi,p)&:=\frac{\partial I}{\partial \pi}(\pi,p)=\int_{(-1,\infty)}\left(\frac{x}{(1+\pi x)^p}-x-\frac{px^2}{(1+\pi x)^{1+p}}\right)\,dF(x),\\
    M(\pi,p)&:=\frac{\partial I}{\partial p}(\pi,p)=\int_{(-1,\infty)}\frac{-x\pi}{(1+\pi x)^p}\log(1+x\pi)\;dF(x),
  \end{align*} 
  are continuous for all $\pi\in(\mathcal{A}_{0,*})^\circ\backslash\{0\}$ and $p\in(p_a,p_b)$.
\end{lemma}
\begin{proof}
  We denote the integrand of $K$ by $k(x,\pi,p)$, and the integrand of $M$ by $m(x,\pi,p)$.

  Let $\pi\in(\mathcal{A}_{0,*})^\circ\backslash\{0\}$, $p\in(p_a,p_b)$ and assume first $\pi>0$. Let $\pi_n\in(\mathcal{A}_{0,*})^\circ$, $p_n\in(p_a, p_b)$ with $\pi_n>0$ and $\pi_n\to \pi$, $p_n\to p$. Letting $\underline \pi:=\inf_n \pi_n$, $\overline\pi:=\sup_n\pi_n$, $\underline p:=\inf_n p_n$, $\overline p:=\sup_n p_n$, we find the following estimates  which hold $F$-a.s.:
   \begin{align*}
      |k(x,\pi_n,p_n)|&\leq |x|(1+\overline \pi x)^{-\overline{p}}+|x|+\overline p x^2(1+\overline \pi x)^{-1-\overline p}\quad \text{for } -1<x<0,\\
      |k(x,\pi_n,p_n)|&\leq x(1+\underline \pi x)^{-\underline{p}}+x+\overline p x^2(1+\underline \pi x)^{-1-\underline p}\quad \text{for } x\geq0,\\
      |m(x,\pi_n,p_n)|&\leq \left|\overline \pi x\log(1+\overline \pi x)\right|(1+\overline\pi x)^{-\overline p}\quad \text{for } -1<x<0,\\
      |m(x,\pi_n,p_n)|&\leq \underline{p}^{-1}\overline \pi x\quad \text{for } x\geq 0,
   \end{align*}
  where we have used in the last inequality the estimate $\log(1+x\pi_n)\leq p^{-1}(1+x\pi_n)^p$ for all $x\geq 0$, $p>0$. We claim that the right hand side estimates are $F$-integrable functions, which implies by dominated convergence that $K$ and $M$ are continuous for $\pi>0$: this follows from our assumption that $L$ is integrable as well as that the integral of the functions on the right hand sides can be considered as derivatives of component-wise  convex (concave) functions, as shown in the proof of Proposition \ref{prop:levy_order_of_strategies}. The integrals are finite since $\overline p,\underline p,\overline \pi,\underline \pi$ are in the interior of $(p_a,p_b)$ and $\mathcal{A}_{0,*}$, respectively.

  Letting $\pi\in(\mathcal{A}_{0,*})^\circ\backslash\{0\}$ with $\pi<0$, $p\in(p_a,p_b)$ and $\pi_n\in(\mathcal{A}_{0,*})^\circ$, $p_n\in(p_a,p_b)$ with $\pi_n<0$ and $\pi_n\to \pi$, $p_n\to p$, we find with the same notation as above
  \begin{align*}
      |k(x,\pi_n,p_n)|&\leq |x|(1+\overline \pi x)^{-\overline{p}}+|x|+\overline p x^2(1+\overline \pi x)^{-1-\underline p}\quad \text{for } -1<x<0,\\
      |k(x,\pi_n,p_n)|&\leq x(1+\underline \pi x)^{-\underline{p}}+x+\overline p x^2(1+\underline \pi x)^{-1-\overline p}\quad \text{for } x\geq0,\\
      |m(x,\pi_n,p_n)|&\leq \underline{p}^{-1}|x||\underline \pi|\quad \text{for } -1<x<0,\\
      |m(x,\pi_n,p_n)|&\leq \left|\underline \pi x \log(1+\underline \pi x)\right|(1+\underline \pi x)^{-\overline p}\quad \text{for } x\geq 0.
   \end{align*}
  Similar arguments as above show that the right hand side estimates are $F$-integrable. Hence, $K$ and $M$ are also continuous for $\pi<0$. 
\end{proof}


\bibliography{../bibliography}{}

\begin{thebibliography}{10}

\bibitem{BaUr10}
N.~B{\"a}uerle, S.P. Urban, and L.A.M. Veraart.
\newblock The relaxed investor with partial information.
\newblock {\em Preprint}, 2010.

\bibitem{benth_karlsen01optimal}
F.E. Benth, K.~H. Karlsen, and K.~Reikvam.
\newblock Optimal portfolio selection with consumption and nonlinear
  integro-differential equations with gradient constraint: a viscosity solution
  approach.
\newblock {\em Finance and Stochastics}, 5(3):275--303, 2001.

\bibitem{figueroa08small}
J.E. Figueroa-L{\'o}pez.
\newblock {Small-time moment asymptotics for L{\'e}vy processes}.
\newblock {\em Statistics \& Probability Letters}, 78(18):3355--3365, 2008.

\bibitem{Fo89}
L.~Foldes.
\newblock Certainty equivalence in the continuous-time portfolio-cum-saving
  model.
\newblock In {\em Applied stochastic analysis ({L}ondon, 1989)}, volume~5 of
  {\em Stochastics Monogr.}, pages 343--387. Gordon and Breach, New York, 1991.

\bibitem{framstad_oksendal99optimal}
N.C. Framstad, B.K. {\O}ksendal, and A.~Sulem.
\newblock {Optimal consumption and portfolio in a jump diffusion market}.
\newblock {\em Journal of Mathematical Economics}, 1999.

\bibitem{goll_kallsen00}
T.~Goll and J.~Kallsen.
\newblock {Optimal portfolios for logarithmic utility}.
\newblock {\em Stochastic Processes and their Applications}, 89(1):31--48,
  2000.

\bibitem{He91}
H.~He.
\newblock Optimal consumption-portfolio policies: a convergence from discrete
  to continuous time models.
\newblock {\em J. Econom. Theory}, 55(2):340--363, 1991.

\bibitem{jacod07asymptotics}
J.~Jacod.
\newblock Asymptotic properties of power variations of {L}{\'e}vy processes.
\newblock {\em ESAIM. Probability and Statistics}, 11:173--196, 2007.

\bibitem{jacod_shiryaev03limit}
J.~Jacod and A.N. Shiryaev.
\newblock {\em {Limit Theorems for Stochastic Processes}}.
\newblock Springer Berlin, second edition, 2003.

\bibitem{kallsen00optimal}
J.~Kallsen.
\newblock Optimal portfolios for exponential {L}{\'e}vy processes.
\newblock {\em Mathematical Methods of Operations Research}, 51(3):357--374,
  2000.

\bibitem{KaKa07}
I.~Karatzas and C.~Kardaras.
\newblock The num\'eraire portfolio in semimartingale financial models.
\newblock {\em Finance Stoch.}, 11(4):447--493, 2007.

\bibitem{kardaras09no}
C.~Kardaras.
\newblock {No-Free-Lunch equivalences for exponential L{\'e}vy models under
  convex constraints on investment}.
\newblock {\em Mathematical Finance}, 19(2):161--187, 2009.

\bibitem{kohatsu_protter94euler}
A.~Kohatsu and P.E. Protter.
\newblock {The euler scheme for SDE's driven by semimartingales}.
\newblock {\em Stochastic Analysis on Infinite Dimensional Spaces}, pages
  141--151, 1994.

\bibitem{kramkov_schachermayer99}
D.~Kramkov and W.~Schachermayer.
\newblock {The asymptotic elasticity of utility functions and optimal
  investment in incomplete markets}.
\newblock {\em Annals of Applied Probability}, 9(3):904--950, 1999.

\bibitem{merton69lifetime}
R.C. Merton.
\newblock {Lifetime portfolio selection under uncertainty: The continuous-time
  case}.
\newblock {\em The Review of Economics and Statistics}, 51(3):247--257, 1969.

\bibitem{nutz09power}
M.~Nutz.
\newblock {Power utility maximization in constrained exponential L{\'e}vy
  models}.
\newblock {\em Mathematical Finance (forthcoming)}, 2009.

\bibitem{nutz12bellman}
M.~Nutz.
\newblock The bellman equation for power utility maximization with
  semimartingales.
\newblock {\em The Annals of Applied Probability}, 22(1):363--406, 2012.

\bibitem{RoSt02}
L.~C.~G. Rogers and E.~J. Stapleton.
\newblock Utility maximisation with a time lag in trading.
\newblock In {\em Computational methods in decision-making, economics and
  finance}, volume~74 of {\em Appl. Optim.}, pages 249--269. Kluwer Acad.
  Publ., Dordrecht, 2002.

\bibitem{rogers01relaxed}
L.C.G. Rogers.
\newblock {The relaxed investor and parameter uncertainty}.
\newblock {\em Finance and Stochastics}, 5(2):131--154, 2001.

\bibitem{samuelson69lifetime}
P.A. Samuelson.
\newblock {Lifetime portfolio selection by dynamic stochastic programming}.
\newblock {\em The Review of Economics and Statistics}, 51(3):239--246, 1969.

\bibitem{sato99levy}
K.I. Sato.
\newblock {\em {L{\'e}vy Processes and Infinitely Divisible Distributions}}.
\newblock Cambridge University Press, 1999.

\end{thebibliography}
\bibliographystyle{plain}

\end{document}